\documentclass[11pt]{article}
\usepackage{amssymb,amsmath,amsthm,latexsym}
\usepackage{verbatim}
\usepackage{graphicx}

\newtheorem{theorem}{Theorem}
\newtheorem*{property*}{Property}
\newtheorem*{definition*}{Definition}

\begin{document}
\title{\textbf{Normality of the Ehrenfeucht-Mycielski Sequence}}
\author{
\\Kundan Krishna\\Department Of Computer Science And Engineering,\\Indian Institute of Technology,Kanpur,\\KANPUR (UP)-208016\\Email:kkrishna@iitk.ac.in
\and
\\Mentor\\Dr. Satyadev Nandakumar,\\Department Of Computer Science And Engineering \\\\\\
}
\date{July 9, 2014}

\maketitle
\textit{\\\\\\}
\abstract{We study the binary Ehrenfeucht Mycielski sequence seeking a balance between the number of occurrences of different binary strings. There have been numerous attempts to prove the balance conjecture of the sequence, which roughly states that 1 and 0 occur equally often in it. Our contribution is twofold. First, we study weaker forms of the conjecture proved in the past and lay out detailed proofs for many lemmas which were stated without proofs. Secondly, we extend the claim of balance to that of normality and prove a weaker form of simple normality to word length 2.}

\newpage

\tableofcontents

\newpage

\section{Introduction}
In [3] Ehrenfeucht and Mycielski propose a pseudorandom binary sequence, henceforth
called the EM sequence.The EM-sequence is sequence A038219 in the encyclopedia (Slo07), and is generated via an algorithm as follows: “The sequence starts
0,1,0 and continues according to the following rule: find the longest sequence at the end that has occurred
at least once previously. If there are more than one previous occurrences select the last one. The next digit
of the sequence is the opposite of the one following the previous occurrence.” For example, the first 30
terms of the EM-sequence are
 $$010011010111000100001111011001$$
 The longest suffix occurring before is 1001 as shown
 $$0\textbf{1001}101011100010000111101\textbf{1001}$$
 Since the first 1001 was followed by 1, the 31st bit will be 0.\\

The EM sequence arises from the study of decision method used by all
learning organisms. Suppose that we follow the procedure outlined in the above
paragraph, but instead of flipping the bit following the penultimate occurrence
of the sequence, we take it as is; this may be viewed as making our decision
based on past experience, where we search for the event in the past that looked
most like our current predicament.From this point of view, if we flip the last bit (instead of taking it as
is) we are simulating the case in which we are always wrong. 

\subsection{The balance conjecture and normality}
The famous balance conjecture of the EM-sequence states that as the number of bits in the sequence grows towards infinity, the ratio of number of 0s and 1s in the sequence will converge to 1.Coined in 1992, the conjecture remains unsolved till date. Weaker variations of the balance conjecture have been proved, the best of which was proved by Keiffer and Szpankowski in \cite{keiffer}. Unfortunately, the paper only gave an outline of the method and did not give proofs. In this paper, we prove the theorems given in Keiffer's paper.
We analysed the first few thousand bits of the sequence, which suggested that it is normal sequence as well. We extend the method in \cite{keiffer} to prove a weaker form of simple normality of length 2 binary strings, and explore the feasibility of the method for longer strings.
\section{Notations and Previous works from literature}
\subsection{Terminology}
We specify the notation and terminology that will remain in force throughout the paper. $\{0,1\}^+$ denotes the set of all binary strings of finite nonzero length. $\lambda$ denotes the empty string. Capital letters are used to represent strings while small letters for representing bits. String uv is the (left-to-right)
concatenation of string u with string v. $x_i^j$ and em[i..j] denote the
substring of EM-sequence from the $i^{th}$ to $j^{th}$ bit. $N_n(w)$ denotes the number of occurrences of w in $x_1^n $ and $N_k^n(w)$ denotes the number of occurrences of w in $x_k^n$. $|B|$ denotes the length of string $B \in \{0,1\}^+$ . If $a \in \{0,1\}$, $\bar{a}=1-a$  .card(S) or $|S|$ denotes the cardinality of set S. $|T|$ denotes the number of vertices of
tree T .

We define $\alpha(n)$ as the maximum of the lengths of all the strings matched during the formation of the first n bits of the sequence. A string B is called good if the first two occurrences of B have different bits just before them, and bad otherwise.

\subsection{Previous work}
Ehrenfeucht and Mycielski in \cite{em} proved that every possible binary word of any length occurs infinite number of times in the sequence. Following this, in \cite{mcconell},an attempt was made to prove the balance conjecture which yielded that minimum 0.02 fraction of the bits will be 1s. In  \cite{sutner}, Sutner bound this ratio in the interval [0.11,0.89]. The best result is in \cite{keiffer} where this bound is further tightened to [0.25,0.75].Seeing the novel approach by Keiffer and Szpankowski, we decided to explore it to find an improvement.

\subsection{Definitions and used results}
\textit{Definitions. }
\begin{itemize}
\item For each positive integer n, we define $R_n$ to be the set consisting of those strings in $\{0,1\}^*$ which occur at least twice as substrings of the initial segment $x_1^n$ of the EM-sequence.
\item We construct $T_n$, the tree made with exactly the strings in $R_n$ with bits as edgelabels such that concatenated edgelabes of each path from a vertex to root corresponds to a string in $R_n$.\footnote{see \cite{keiffer} for more information on $T_n$ and $T_n^*$}
\item we define $b^+(i)$ to be the longest string starting at $i$ which appears as a substring of the EM-sequence for at least the second time starting at position i, and we define $L_i$ to be the length of $b^+(i)$.
\item Simple Normality is 
\end{itemize}
\textit{Results. }
\begin{itemize}
	\item Whenever a new maxima is attained in the matchlengths during the formation of EM-sequence, the matched string is and initial string($x_1^i$ for some $i$).\footnote{proved in \cite{mcconell}}
	\item For each positive integer k, let $i_k$ be the integer such that $b^+(i_k) = x_1^k$ . Then there is a positive constant C such that $i_k \geq C(2^{k/2} ), k \geq 1$.\footnote{proved in \cite{sutner}} Together with the previous result, it implies that the matchlengths will grow as $\Theta(\log n)$
\end{itemize}

\newpage
\section{Our work}
Our contribution is threefold.
\begin{itemize}
	\item The first part is a collection of proofs for the results given in \cite{keiffer}.
	\item The second part deals with normality of the sequence using the extension of the idea of balanced number of 0s and 1s to balanced number of binary words of length 2 and more.
	\item We wrote a set of programs in Mathematica to calculate the EM-sequence, extract the set $R_n$ and the graphically render the tree $T_n$ with color based differentiation of the balanced and unbalanced parts. The code can be found here: https://bitbucket.org/kukrishna/ehrenfeucht-mycielski-sequence
\end{itemize}

\subsection{Proofs of theorems in Keiffer's paper}
Please note that the indices of the propositions and theorems are according to their appearance in \cite{keiffer}.
\subsubsection{Proofs for Proposition 3.1}

\subsubsection*{Proposition 3.1a}
$$|R_{n}|=n+o(n)$$
\begin{proof}

If $A \in b^{+}(i)$ then surely since A is occurring for the second time, $A	\in R_{n}$.

We now show that the converse is also true.\\
Assume that $A	\in R_{n}$ and $A \notin b^{+}(i)$ for $1<i<n$. Suppose, A occurs for the second time starting at position k. So $b^{+}(k)$ should not be A. And hence, $b^{+}(k)=AB$ where B is a string following A of nonzero length. So AB occurred for the second time at k as well. So, the first occurrence of A must also be followed by B. This is not possible because the first two occurrences of all finite words are followed by different bits.

Hence, $A	\in R_{n} \implies A \in b^{+}(i)$ for some $1<i<n$.

We conclude that $R_{n}$ contains exactly the set of strings $$\{b^{+}(i):1\leq i \leq x\}$$
where $x$ is the largest index such that $x+L_{x}-1 \leq n$. Note that $R_{n}$ also trivially contains the null string. \\
Also, $L_{x} \leq \alpha(n) \\ \implies x \geq n-\alpha(n)+1$

All $b^{+}(i)$s are unique except for $b^{+}(0)$ and $b^{+}(1)$, both of which equal $\lambda$.So, $$|R_{n}|=(x-1),\; \forall n \geq 2 \footnote{Mostly $x$ remains close to $n-\alpha(n)+1$,and hence $|R_{n}|$ closely follows $n-\alpha(n)$.Eg $\alpha(1000)=13$ and $|R_{1000}|=987$} $$
\end{proof}

\subsubsection*{Propositions 3.1b and 3.1c}

$$ card(\{b \in R_{n} : \text{0 is rightmost bit of b}\}) = N_{n}(0)+o(n)$$
$$ card(\{b \in R_{n} : \text{1 is rightmost bit of b}\}) = N_{n}(1)+o(n)$$

\begin{proof}
First we show that any string which is matched during the formation of the EM sequence must be there in $R_{n}$.

Suppose at time j, em[i..j] is the matched string. Then, we claim that $b^{+}(i)$ will be  $em[i..j]$ in all except two cases. In $\alpha(n)$ cases the the matched string is an initial string or and has its third occurrence. $em[i..j]$ is matched at time j, but $b^+(i)=em[i..j+1]$ in those cases. Also, if a bad word is matched as $em[i..j]$, it is its third occurrence and here also $b^+(i)=em[i..j+1]$. We observe that in each of the previous two cases, the matched string under consideration is already included in $R_n$ as $b^+(i)$ where i marks the beginning of its second occurrence.

In all the rest cases, the matched strings are their second occurrences. As we know that the first two occurrences of any string are followed by complementary bits, $em[i..j+1]$ does not occur before,and hence $b^{+}(i)=em[i..j]$.

At each time t, suppose $em[t]=b$,then we have a matched string ending at t and ending with a b. So this will be in $R_{n}$. Amongst the n matches in the first n  bits,we have $\alpha(n)$ initial strings matching twice, while the rest $n-\alpha(n)$ matched strings are all distinct, while  all belong to $R_{n}$. So, 
\begin{equation}
card(\{b \in R_{n} : \text{0 is rightmost bit of b}\})\geq N_{n}(0)-N_{\alpha(n)}(0)\footnote{Note that $\alpha(n)$ is o(n) }
\end{equation}
\begin{equation}card(\{b \in R_{n} : \text{1 is rightmost bit of b}\})\geq N_{n}(1)-N_{\alpha(n)}(1)
\end{equation}

\begin{equation} card(\{b \in R_{n} : \text{0 is rightmost bit of b}\})+ card(\{b \in R_{n} : \text{1 is rightmost bit of b}\}) = |R_{n}|
\end{equation}

Suppose Proposition 3.1b is false. Then 
\begin{equation}
card(\{b \in R_{n} : \text{0 is rightmost bit of b}\}) \geq N_{n}(0)+cn
\end{equation}
for some constant c and infinitely many values of n. Let the set of all such n's be S. 
Adding equation 2 and 4, and substituting into equation 3 we see 
\begin{equation}
	|R_{n}| \geq N_{n}(0)+cn+N_{n}(1)-N_{\alpha(n)}(1)		,\; \forall n \in S
\end{equation}
But $N_{n}(1)+N_{n}(0)=n$, and so 5 reduces to,
$$
	|R_{n}| \geq n+cn-N_{\alpha(n)}(1)		,\; \forall n \in S
$$
Since $N_{\alpha(n)}(1)$ is o(n), $\exists b<c$, such that,
$$
	|R_{n}| \geq n+bn		,\; \forall n \in S
$$

This contradicts proposition 3.1a and hence,
\begin{equation*}
card(\{b \in R_{n} : \text{0 is rightmost bit of b}\}) = N_{n}(0)+o(n)
\end{equation*}
Similarly, by symmetry
\begin{equation*}
card(\{b \in R_{n} : \text{1 is rightmost bit of b}\}) = N_{n}(0)+o(n)
\end{equation*}

\end{proof}

\subsubsection{Proofs of Lemmas 4.1 to 4.4}
In order to prove these lemmas, we coin a new tool called proximty.
\subsection*{Proximity}
\begin{definition*}
Proximity between two strings A and B is defined as the length of their maximal common prefix. The maximal common prefix is called the proximal word.
\end{definition*}
Find $\sigma$ such that they occur as  $b\sigma A \ldots \bar{b}\sigma B$, then $|\sigma|$ is the proximity between A and B. We denote proximity between A and B by p(A,B) and the prefix by P(A,B). We will use the following concept for proving the lemmas.

\begin{property*}
If $p(X,Y) \neq p(X,Z)$,then $p(Y,Z) = min(p(X,Y), p(X,Z))$
\end{property*}
The proof is easy to see. Let $ p(X,Y) < p(X,Z)$ . Let $P(X,Z)=\sigma$. Then, P(X,Y) is a proper suffix of $\sigma$, let's say it's $\beta$ such that $b\beta$ occurs just before Y. If $b\beta$ were a suffix of $\sigma$, P(X,Y) would be $b\beta$. But that can't happen, and as $\sigma$ also occurs before Z, so $P(Y,Z)=P(X,Y)=\beta$.

\subsubsection*{Lemma 4.1}
$B \in \{0,1\}^{+}$ .Then the first 5 occurrences of B in the em-sequence cannot take the form $B\bar{a}_{1},Ba_{1}a_{2},$ $Ba_{1}\bar{a}_{2}, Ba_{1}a_{2}, Ba_{1}a_{2}$ for some $a_{1},a_{2}\in \{0,1\}$.
\begin{proof}
We refer to the $k^{th}$ term given in the lemma's sequence by using the roman numerals for k. Eg. (ii) for the second term (In this case $ Ba_{1}a_{2}$).
For $a_{1}$ to come after the third B, it must match with first B . Hence, p(i,iii) must be greater that p(ii,iii). By the above property, p(ii,iii)=p(i,ii) and $p(i,iii) > p(i,ii)$. Proceeding in the same way,
$$p(i,v)>p(i,iv)>p(i,iii)>p(i,ii)$$
Hence, $p(iv,v)=p(i,iv)$ and $p(iii,v)=p(i,iii)$\\
$\implies p(iv,v)>p(iii,v)$
So,the $Ba_{1}$ in (v) always matches with (iv), and (v) would surely be $ Ba_{1}\bar{a}_{2}$, not $Ba_{1}a_{2}$. The required pattern cannot be produced.
\end{proof}

\subsubsection*{Lemma 4.2}
$B \in \{0,1\}^{+}$ be good.Then the first 4 occurrences of B in the em-sequence cannot take the form $Ba_{1}\bar{a}_{2}, B\bar{a}_{1}, Ba_{1}a_{2}, Ba_{1}a_{2}$ for some $a_{1},a_{2}\in \{0,1\}$.
\begin{proof}
As B is good, $p(i,ii)=0$. As 0B and 1B both have occurred. Without loss of generality,let 1 precede (ii). The third B matches with the second one,which means (iii) must have 1 just before it. Hence, $p(ii,iii) \geq 1$. Now, in the same fashion as in lemma 4.1,
$$p(ii,iv)>p(ii,iii)$$
$p(iii,iv)=p(ii,iii)\geq 1$,which means that (iv) also has a 1 before it. Effectively making $p(i,iv)=0$. Hence, (iv) matches with (iii) and would give $Ba_{1}\bar{a}_{2}$ instead.

\end{proof}

\subsubsection*{Lemma 4.3}
$B \in \{0,1\}^{+}$ .Then the first 5 occurrences of B in the em-sequence cannot take the form $Ba_{1}a_{2},B\bar{a}_{1},$ $Ba_{1}\bar{a}_{2}, Ba_{1}a_{2}, Ba_{1}a_{2}$ for some $a_{1},a_{2}\in \{0,1\}$.
\begin{proof}
As in the previous case,
$$p(ii,v)>p(ii,iv)>p(ii,iii)$$
$p(iv,v)=p(ii,iv)$ and $p(iii,v)=p(ii,iii)$\\
$\implies p(iv,v)>p(iii,v)$
Again forcing the last term to be $ Ba_{1}\bar{a}_{2}$, instead of $Ba_{1}a_{2}$.
\end{proof}

\subsubsection*{Lemma 4.4}
$B \in \{0,1\}^{+}$ .Then the first 4 occurrences of B in the em-sequence cannot take the form $B\bar{a}_{1}, Ba_{1}\bar{a}_{2},$ $Ba_{1}a_{2}, Ba_{1}a_{2}$ for some $a_{1},a_{2}\in \{0,1\}$.
\begin{proof}
Same as lemma 4.3. As I remarked, first term of lemma 4.3 is absent in 4.4 and still everything works fine.
\end{proof}

\subsubsection{Proof of Proposition 4.1}
Let $B \in \{0,1\}^{+}$ be good word.Let $a<b<c<d<e$ be the positive integers at which the first five occurrences of B in the em-sequence {$x_{i}:i \geq 1$} end. Let $u,v$ be the strings.
$$ u=x_{a+1}x_{b+1}x_{c+1}x_{d+1}x_{e+1},\quad v=x_{a+2}x_{b+2}x_{c+2}x_{d+2}x_{e+2},$$

Then at least one of the following statements must be true:\\
(a): $|N_{u}(0)-N_{u}(1)| \leq 1$\\
(b): $|N_{v}(0)-N_{v}(1)| \leq 1$

\begin{proof}
Consider any arbitrary B.If condition (a) fails, the first five occurrences will be either $ B\bar{a}_{1},Ba_{1},Ba_{1},Ba_{1},Ba_{1}$ (as in lemma 4.1 and 4.4) or $Ba_{1},B\bar{a}_{1},Ba_{1},Ba_{1},Ba_{1}$ (lemma 4.2 and 4.3).\\
\textbf{Case 1 : $ B\bar{a}_{1},Ba_{1},Ba_{1},Ba_{1},Ba_{1}$}\\
The second and third occurrence of $Ba_{1}$ will always be followed by complementary bits $a_{2}$ and $\bar{a}_{2}$. For (b) to fail, either $B\bar{a}_{1}, Ba_{1}a_{2}, Ba_{1}\bar{a}_{2}, Ba_{1}a_{2}, Ba_{1}a_{2}$ or $B\bar{a}_{1}, Ba_{1}a_{2}, Ba_{1}\bar{a}_{2}, Ba_{1}\bar{a}_{2}, Ba_{1}\bar{a}_{2}$ must be the case. But both these possibilities are rejected by lemmas 4.1 and 4.4 respectively.\\
\textbf{Case 2 : $Ba_{1},B\bar{a}_{1},Ba_{1},Ba_{1},Ba_{1}$}\\
For (b) to fail,either $Ba_{1}a_{2},B\bar{a}_{1}, Ba_{1}\bar{a}_{2}, Ba_{1}a_{2}, Ba_{1}a_{2}$ or $Ba_{1}a_{2},B\bar{a}_{1}, Ba_{1}\bar{a}_{2},Ba_{1}\bar{a}_{2},Ba_{1}\bar{a}_{2}$ must occur. First possibility is discarded by lemma 4.3 . Second possibility is discarded by lemma 4.2 if B is known to be bad.
\end{proof}

\newpage
\subsubsection{Proof of Corollary 4.1}
For each n, the set of all edges of $T_{n}$ which belong to spaghetti strands may be partitioned into two subsets $E_n (1)$, $E_n (2)$ satisfying the following properties:
\begin{itemize}
\item For each n, $E_n (1)$ contains at most 2 edges from each spaghetti strand of $T_n$ .
\item $|E_n (2)| = o(n)$.
\end{itemize}

\begin{proof}

Suppose that we have a spaghetti strand which has more than 2 edges as shown in the figure.

\begin{figure}[h]
\centering
\includegraphics[width=150px,height=150px]{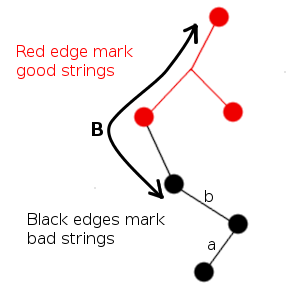}
\end{figure}

In the above figure the B is the address of the word which belongs to, say $T_n(0)$ and not $T_n(1)$. So, $B0 \in R_n$ while $B1 \notin R_n$. Also it has a length 2 strand attached below it. According to the figure, they are $bB \in T_n(0)$ and $abB \in T_n(0)$ (and not in $T_n(1)$).All of $abB0,bB0$ and $B0$ are good strings.

Since, the first two occurrences of any word are followed by complementary bits, we shall have atleast 3 occurrences of $abB$ as $abB0,abB1,abB0$. Since bB0 is a good string, there must be a $\bar{a}bB0$ occurring before the second $abB0$. Similarly since B is good there must also be a $\bar{b}B0$ before the second $abB0$. Hence we see that $B$ must occur atleast 5 times before $abB0$ is inducted into $R_n$.This concludes the first part of the proof.

Any word B which occurs 5 times, will follow this - either $ B \in T_n^*$, or $Bb \in T_n^*$ for some b. 

If the first 5 occurrences of B are followed by atleast two each of $a_1$ and $\bar{a}_1$, then $B \in T_n^*$. Otherwise the following cases may occur-

a) If the first 5 occurrences are $ B\bar{a}_{1},Ba_{1},Ba_{1},Ba_{1},Ba_{1}$, then the proximity of the words with (i) will strictly increase from (ii) to (v).So, every $Ba_{1}$
will match with the one immediately before it.Thus forming $ B\bar{a}_{1},Ba_{1}a_{2},Ba_{1}\bar{a}_{2},Ba_{1}a_{2},Ba_{1}\bar{a}_{2}$.

b) If the first 5 occurrences are $ Ba_{1},B\bar{a}_{1},Ba_{1},Ba_{1},Ba_{1}$,then the first three terms are given by $ Ba_{1}a_{2},B\bar{a}_{1},Ba_{1}\bar{a}_{2}$. After that since B is good and the Bs of (iii),(iv),(v) are all matched with (ii), $p(i,x)=0$ for $x \in \{iii,iv,v\}$. We also see $p(ii,iii)<p(ii,iv)<p(ii,v)$. This would lead to fourth $Ba_{1}$ matching with (iii) and fifth $Ba_{1}$ matching with (iv). Thus leading to $Ba_{1}a_{2},B\bar{a}_{1},Ba_{1}\bar{a}_{2},Ba_{1}a_{2},Ba_{1}\bar{a}_{2}$

Hence, in both cases a and b $Ba_{1} \in T_n^*$.

Now we make a sharp observation to conclude the proof of Corollary 4.1 . If $B\in Tn*$, we are done. But if $Ba_{1}\in T_n^*$,by virtue of the goodness of B and the pattern of proximity shown earlier, the first 5 occurrences of B will be exactly one of the following two cases-
$$ cB\bar{a}_{1},\bar{c}Ba_{1}a_{2},cBa_{1}\bar{a}_{2},cBa_{1}a_{2},cBa_{1}\bar{a}_{2}\quad \text{   or   } \quad \bar{c}Ba_{1}a_{2},cB\bar{a}_{1},cBa_{1}\bar{a}_{2},cBa_{1}a_{2},cBa_{1}\bar{a}_{2}$$

In both the cases, first two occurrences of $Ba_{1}\bar{a}_2$ is preceded by c, hence making it a bad string which belongs in $R_n$.

For each spaghetti strand with all good strings and $length>2$, we shall have a different string serving as B and hence a different bad string yielding out of it.Since there are only o(n) bad strings, we conclude that this case can only arise o(n) times. Think that you have a basket full of apples(bad strings). Each time you encounter a good string(here B) occurring 5 or more times and also unbalanced(ie. $B0 \in R_n$ and $B1 \in R_n$), you have to pick one out of it. There is also the other case where B is not good.In that case, if $B$ is unbalanced, take $B=xyZ$, and repeat the whole process where abB is replaced by xyZ.Going from B to Z will cost you two more black edges. So, the maximum total number of extra black edges that you might have to spend in such transitions is $2o(n)$ in the whole tree. After that you will definitely find a good string in B's role, and you will pick out one apple. The basket will become empty in a maximum of $2o(n)+o(n)=3o(n)$ black edges, in addition to the 2 edges in the lowermost part per strand (here marked as a and b).
\end{proof}

\subsubsection{Directions for future work}
Though we managed to prove all the previous theorems, we were unable to prove Proposition 3.1d which states that the number of bad strings in EM-sequence will grow at a rate of $o(n)$. This result, though given amongst the other parts of proposition 3.1, does not seem to be provable by using them and looks more like an independent theorem.However, experimental data agrees with the claim of the proposition. 

\newpage
\subsection{Extension of the method to weaker normality}

\subsubsection{Normality}
   For each a in $\{0,1\}$ let $N_S(a, n)$ denote the number of times the letter a appears in the first n digits of the binary sequence S. We say that S is simply normal if the limit

    $$\lim_{n\to\infty} \frac{N_S(a,n)}{n} = \frac{1}{2}$$

for each a. Now let w be any finite string in $\{0,1\}^+$ with $|w|=l$ and let $N_S(w,n)$ to be the number of times the string w appears as a substring in the first n digits of the sequence S. (For instance, if S = 01010101..., then $N_S(101, 8) = 3$.) S is simply normal to word length $l$ if, for all finite binary strings $w \in \{0,1\}^+$ having length $l$,

   $$\lim_{n\to\infty} \frac{N_S(w,n)}{n} = \frac{1}{2^{l}} $$

  A sequence will be called normal, if it is simply normal to all word lengths.
\subsubsection{Work done}

In this section, we extend the method used in \cite{keiffer} to prove a weaker form of simple normality of the EM-sequence to word length 2. Our primary result establishes 
$$ \limsup\limits_{n\rightarrow \infty}\frac{N_n(x)}{n} \leq \frac{8}{11}\text{ and } \liminf\limits_{n\rightarrow \infty}\frac{N_n(x)}{n} \geq \frac{1}{25} \quad	\forall x \in \{00,01,10,11\}$$

\begin{proof}
Let $T_n(x)$ be the subtree of $T_n$ ending with $x$. Eg. $T_n(01)$ is the subtree of $T_n$ consisting of union of all the paths in $T_n$ ending with 01, corresponding with $R_n(x)$ which is the set of all words in $R_n$ ending with $x$. Similarly, let $N_n(x)$ be the number of occurrences of $x$ in $x_1^n$.

\begin{theorem} 
$|T_n(x)| = |R_n(x)| = N_n(x) + o(n)$
\end{theorem}
\begin{proof}
We know,
 $$R_n=\{b^{+}(i):1 \leq i \leq x\}$$
where $x$ is the largest index such that $x+L_{x}-1 \leq n$. 
All the matched strings in the above range of i will be there in $R_n$. Let the length of $x$ be $l$. As any string is matched only once, and twice if it is an initial string, we will have atmost $2^{l-1}$ words of length $l-1$ in matched. And we also have one word of each length as an initial string, leading to $l-1$ matches happening twice. The number of matches of length $\geq l$ in $x_1^n$, $n_l$ follows-
$$ n_l \leq 2^{l-1}+(l-1)$$
Let $k$ be the index at which $n_l$ achieves the above upper bound. The rest of the matches in $x_k^n$ are of length $l$ or more and so will end with a member of word of length $l$. Let $U_l$ be the set of all binary strings of length $l$.  Analogous  to equation(1) and (2) we shall have,
\begin{equation}
card(\{b \in R_{n} : \text{y is rightmost bit of b}\}) \geq N_k^n(y)- N_{\alpha(n)}(y)\quad \forall y\in U_l
\end{equation}
Suppose for some c and $w\in U_l$,
\begin{equation}
card(\{b \in R_{n} : \text{w is rightmost bit of b}\}) \geq N_k^n(w)+cn
\end{equation}
for infinitely many values of n. Let the set of all such n's be S.
Adding equation 8 and all concrete instances of equation 7 $\forall y\in U_l \setminus\{w\}$, we get
\begin{equation}
	|R_{n}| \geq (n-k) + cn - \alpha(n)	+ N_{\alpha(n)}(w)	,\; \forall n \in S
\end{equation}
To get equation 9 we used the following facts-
$$ |R_n| > card(\{b\in R_n\text{with length(b) $\geq l$}\}) =  \sum\limits_{y\in U_l} card(\{b \in R_{n} : \text{y is rightmost bit of b}\})$$
$$ \sum\limits_{y\in U_l}N_k^n(y) = n-k $$
$$ \sum\limits_{y\in U_l}N_{\alpha(n)}(y) = \alpha(n)$$
Clearly, equation (9) translates to 
$$
	|R_{n}| \geq (c+1)n	- o(n)	,\; \forall n \in S
$$
which contradicts proposition 3.1a and hence,we conclude our assumption is wrong. Indeed,
$$|R_n(x)| = N_n(x) + o(n) \quad \forall x$$
\end{proof}

We attempt to find a ratio amongst the number of words ending with 00,10,01 and 11. Given a string $xyBa0 \in R_n$, we may safely assume that $Ba$ will be balanced. Assuming goodness of all encountered strings, we must have the following words in $x_1^n$ -
$$xyBa0,\bar{y}Ba1,\bar{x}yBa1,xyBa1,\bar{y}Ba0,\bar{x}yBa0,xyBa0$$

But,in addition there must also be $\bar{y}B\bar{a}, yB\bar{a} \text{ and } xyB\bar{a}$. Atleast 2 occurrences of $B\bar{a}$ will be followed by same bit, and hence, $B\bar{a}$ will be in atleast one of $T_n(0)$ or $T_n(1)$. Extending the concept, if we consider $B=pqZ$, then you will also need to have $\bar{q}Z\bar{a}$ and $\bar{p}qZ\bar{a}$ in $x_1^n$. This would make the number of occurrences of $Z\bar{a}$ atleast 5. We know from the proof of Corollary 4.1 that any string which occurs 5 times is either balanced, or if it is unbalanced it costs a bad string. Therefore we conclude that all except o(n) such $Z\bar{a}$s will be balanced. Given that information, we now know that 
$$\{Za1, Za0, Z\bar{a}1, Z\bar{a}0\} \subseteq R_n$$
Let $\zeta_n$ be the set such that $\zeta_n= \{Z : Za1, Za0, Z\bar{a}1, Z\bar{a}0 \in R_n \}$, and $\zeta_n^*$ be the corresponding subtree in $T_n$.Let $\gamma_n=|\zeta_n^*|$.Let $j_{00}(n)$ be the number of edges whose paths to root end with 00. Define $j_{01}(n),j_{10}(n)\text{ and }j_{11}(n)$ similarly. Hence,
$$|T_n(00)|=\gamma_n+j_{00}(n), \quad |T_n(10)|=\gamma_n+j_{10}(n), \quad |T_n(01)|=\gamma_n+j_{01}(n), \quad |T_n(11)|=\gamma_n+j_{11}(n)$$
$$ |T_n(00)|+|T_n(10)|+|T_n(01)|+|T_n(11)|= 4\gamma_n+j_{00}(n)+j_{01}(n)+j_{10}(n)+j_{11}(n) $$
If $L(\zeta_n^*)$ be the number of leaf vertices of $\zeta_n$ and let $U(\zeta_n^*)$ be the number of unary vertices of $\zeta_n^*$. Then,
$$ \gamma_n = 2L(\zeta_n^*) + U(\zeta_n^*) -1$$
The maximum number of edges connected to $\zeta_n^*$ will be $2L(\zeta_n^*)+U(\zeta_n^*)$ which equals $\gamma_n+1$. Assuming maximum possible branching upto depth 2 and susequent spaghetti strands of length 2, we shall have
\begin{equation} 
j_{00}(n) \leq \gamma_n+2 \gamma_n + 4 \gamma_n + o(n)= 7\gamma_n +o(n) 
\end{equation}

A similar result will hold for 10,01 and 11 too.

Moving ahead, from Theorem 1,we have
$$|T_n(00)|+|T_n(10)|+|T_n(01)|+|T_n(11)|=N_n(00)+N_n(01)+N_n(10)+N_n(11)+o(n)=n+o(n)$$
$$|T_n(00)|=N_n(00)+o(n)$$
we shall have,
$$\limsup\limits_{n\rightarrow \infty}\frac{N_n(00)}{n}=\limsup\limits_{n\rightarrow \infty}\frac{|T_n(00)|}{|T_n(00)|+|T_n(10)|+|T_n(01)|+|T_n(11)|}$$
Using equation (10), we can find a sequence of positive numbers $\{\epsilon_n\}$ tending to 0 such that,
$$ j_{00}(n) \leq 7 \gamma_n + n\epsilon_n, \quad n=1,2,3,4,..$$
Using this, we obtain
\begin{align*} 
\frac{|T_n(00)|}{|T_n(00)|+|T_n(10)|+|T_n(01)|+|T_n(11)|} &= \frac{\gamma_n + j_{00}(n)}{4\gamma_n+j_{00}(n)+j_{10}(n)+j_{01}(n)+j_{11}(n)} \\
&\leq \frac{\gamma_n + j_{00}(n)}{4\gamma_n+ j_{00}(n)}\\
&\leq \frac{8\gamma_n + n\epsilon_n}{11\gamma_n +n\epsilon_n} \\
&\leq \frac{8}{11} + \frac{n\epsilon_n}{11\gamma_n}
\end{align*}

Lastly, we note that 
\begin{align*} n=|T_n(00)|+|T_n(10)|+|T_n(01)|+|T_n(11)|+o(n)&=4\gamma_n+j_{00}(n)+j_{10}(n)+j_{01}(n)+j_{11}(n)+o(n) \\
&\leq 32\gamma_n + o(n)
\end{align*}

making the term $n/\gamma_n$ $O(1)$ and hence the ratio will converge to 8/11.

Attemting to find a lower bound, 
\begin{align*} 
\frac{|T_n(00)|}{|T_n(00)|+|T_n(10)|+|T_n(01)|+|T_n(11)|} &= \frac{\gamma_n + j_{00}(n)}{4\gamma_n+j_{00}(n)+j_{10}(n)+j_{01}(n)+j_{11}(n)} \\
&\geq \frac{\gamma_n }{4\gamma_n+ j_{10}(n)+ j_{01}(n)+ j_{11}(n)}\\
&\geq \frac{\gamma_n }{25\gamma_n +n\epsilon_n} \\
&\geq \frac{\gamma_n }{25\gamma_n +32\gamma_n\epsilon_n} \\
&= \frac{1}{25+32\epsilon_n}
\end{align*}
As $\{\epsilon_n\}\rightarrow 0$, the ratio converges to $\frac{1}{25}$.

Hence, we arrive at our final conclusion,
$$ \limsup\limits_{n\rightarrow \infty}\frac{N_n(x)}{n} \leq \frac{8}{11}\text{ and } \liminf\limits_{n\rightarrow \infty}\frac{N_n(x)}{n} \geq \frac{1}{25} \quad	\forall x \in \{00,01,10,11\}$$

\end{proof}

\subsubsection{Directions for future work}
Experimental results show that the EM sequence is simply normal to all word lengths. However, the results obtained by the proof method employed by Keiffer and extended here, will perform worse in higher word lengths. The primary reason for this is that unlike spaghetti strands which are linear, the edges hanging from $\zeta_n^*$ in this case had branches . This will still be the case for higher word lengths which will require us to move up in the tree.The branched extensions to the balanced parts will bring a near exponential increase in the number of unbalanced edges, which will easily outweigh the number of balanced strings by a large margin. Therefore, this method becomes worse for analysis of higher word lengths. The key to obtaining a perfectly balanced ratio lies in finding a ratio between the unbalanced words ending with different words in $U_l$. Any such relation obtained, however small or large, will definitely improve the ratios given by us and \cite{keiffer}.

\addcontentsline{toc}{section}{References}
\bibliography{emsnormality}{}
\bibliographystyle{plain}

\end{document}